\documentclass[11pt]{article}

\usepackage[a4paper,margin=1in]{geometry}
\usepackage[T1]{fontenc}
\usepackage{lmodern}
\usepackage{microtype}

\usepackage{amsmath,amssymb,amsthm,mathtools,mathrsfs}
\usepackage{enumitem}
\usepackage{algorithm}
\usepackage{algpseudocode}

\usepackage{hyperref}
\usepackage[nameinlink,noabbrev]{cleveref}
\usepackage{doi}

\theoremstyle{plain}
\newtheorem{theorem}{Theorem}[section]
\newtheorem{lemma}[theorem]{Lemma}
\newtheorem{proposition}[theorem]{Proposition}
\newtheorem{corollary}[theorem]{Corollary}

\theoremstyle{definition}
\newtheorem{definition}[theorem]{Definition}

\theoremstyle{remark}
\newtheorem{remark}[theorem]{Remark}

\newcommand{\R}{\mathbb{R}}
\newcommand{\Z}{\mathbb{Z}}
\newcommand{\Q}{\mathbb{Q}}
\newcommand{\NN}{\mathcal{N}}
\newcommand{\relu}{\operatorname{ReLU}}

\newcommand{\LRA}{\textsc{LRA}}
\newcommand{\SMT}{\textsc{SMT}}
\newcommand{\MILP}{\textsc{MILP}}

\newcommand{\UNSAT}{\textsc{UNSAT}}

\newcommand{\SMTLRA}{\ensuremath{\SMT(\LRA)}}

\title{Proof-Carrying Verification for ReLU Networks via Rational Certificates}

\author{Chandrasekhar Gokavarapu\\
Department of Mathematics, Government College (Autonomous), Rajahmundry, India\\
\texttt{chandrasekhargokavarapu@gmail.com}}

\date{}

\begin{document}
\maketitle

\begin{abstract}
Rectified Linear Unit (ReLU) networks are piecewise-linear (PWL), so universal linear safety properties can be reduced to reasoning about linear constraints. Modern verifiers rely on \SMTLRA procedures or \MILP{} encodings, but a safety claim is only as trustworthy as the evidence it produces. We develop a proof-carrying verification core for PWL neural constraints on an input domain $D\subseteq\R^n$. We formalize the exact PWL semantics as a union of polyhedra indexed by activation patterns, relate this model to standard exact \SMT{}/\MILP{} encodings and to the canonical convex-hull (ideal) relaxation of a bounded ReLU, and introduce a small certificate calculus whose proof objects live over $\Q$. Two certificate types suffice for the core reasoning steps: entailment certificates validate linear consequences (bound tightening and learned cuts), while Farkas certificates prove infeasibility of strengthened counterexample queries (branch-and-bound pruning). We give an exact proof kernel that checks these artifacts in rational arithmetic, prove soundness and completeness for linear entailment, and show that infeasibility certificates admit sparse representatives depending only on dimension. Worked examples illustrate end-to-end certified reasoning without trusting the solver beyond its exported witnesses.
\end{abstract}

\paragraph{Keywords.}
ReLU networks; piecewise-linear verification; convex-hull semantics; \SMTLRA; \MILP{}; proof-carrying verification;
rational certificates; Farkas lemma.
\paragraph{MSC Classification}{68T07,68V35,90C05,90C11,13P15, 68W30}
\bigskip
\section{Introduction}

Formal verification of neural networks replaces black-box testing by a symbolic decision problem:
does there exist an input in a domain $D\subseteq\R^n$ that violates a safety specification?
For ReLU networks, this question is expressible using piecewise-linear constraints, and can be attacked by
\SMTLRA{} procedures or \MILP{} encodings \cite{EhlersATVA2017,CheungVIPR2017}. The challenge is twofold.
First, the PWL semantics induces exponentially many activation patterns, so scalability depends on strong propagation.
Second, in safety-critical contexts, an \UNSAT\ verdict (``no counterexample exists'') should be accompanied by
evidence that can be independently checked \cite{KatzReluplex2017,KatzMarabou2019}.

\paragraph{Context and related work.}
Complete verification of ReLU networks has been pursued via SMT-style search (e.g., Reluplex/Marabou and related SMT-based methods)
\cite{KatzReluplex2017,KatzMarabou2019,HuangCAV2017Safety} and via exact MILP formulations \cite{EhlersATVA2017,TjengMILP2017}.
On the scalable (incomplete) side, abstract-interpretation and convex-relaxation methods such as
AI$^2$, DeepZ, DeepPoly, and CROWN compute sound bounds by propagation
\cite{GehrSP2018AI2,SinghNeurIPS2018DeepZ,SinghPOPL2019DeepPoly,ZhangNeurIPS2018CROWN}.
These bound computations are often combined with branch-and-bound search to recover completeness
\cite{BunelJMLR2020BaB,WangNeurIPS2021BetaCROWN,BakNFM2021nnenum}.
Practical \SMT{} and \MILP{} backends rely on mature solvers such as Z3 and CVC4
\cite{deMouraBjornerTACAS2008Z3,BarrettCVC4CAV2011} and on standard theory foundations
\cite{BarrettSMTHandbook2009,Schrijver1986}.
Recent work has emphasized proof production for neural-network verification; we contribute a lightweight
LP/Farkas certificate layer that is amenable to symbolic checking \cite{ElboherArxiv2025AbstractionProofProduction}.
For a recent benchmark-centric assessment of verifier performance and complementarities, see \cite{KoenigJMLR2024CriticallyAssessing}.

\paragraph{Computer-algebra perspective.}
Modern verifiers already compute linear bounds, generate linear lemmas, and prune infeasible regions.
Our thesis is that the \emph{mathematical objects underlying these steps are inherently symbolic}:
they can be exported as rational certificates (vectors of multipliers) and checked by exact arithmetic.
This turns neural verification into a proof-carrying pipeline where solvers may be untrusted, but certificates are
small, explicit, and machine-checkable.

\paragraph{What is new (direction, not rephrasing).}
We do not claim novelty in the existence of convex-hull ReLU relaxations or in LP duality.
The novel direction is a \emph{compositional certificate calculus} for PWL neural constraints:
each propagation step emits a certificate; certificates compose across refinements; and certificate checking
is isolated as a symbolic computation problem. This aligns verification with the computer-algebra ethos:
exact manipulation and validation of algebraic objects produced by numeric/LP engines.

\paragraph{Contributions.}
The paper makes the following contributions:
\begin{enumerate}
  \item \textbf{Exact PWL semantics and encodings.}
  We formalize ReLU networks as unions of pattern-indexed polyhedra in extended space and relate this
  semantics to exact \SMT/\MILP{} encodings and to canonical convex-hull relaxations.

  \item \textbf{Proof-carrying certificate calculus over $\Q$.}
  We present a solver-agnostic certificate interface based on rational objects and give exact checkers for
  (i) Farkas infeasibility certificates and (ii) entailment certificates for validated propagation and learned cuts.

  \item \textbf{Computer-algebra normalization and log compression.}
  We introduce symbolic normalization/canonicalization routines for certificates (and, optionally, constraint rows)
  that preserve checkability under exact arithmetic, enabling compact proof logs.

  \item \textbf{Sparse proof artifacts.}
  We prove that infeasibility admits a dimension-sparse rational certificate (support bounded by $d+1$),
  which yields principled compression guarantees for pruning proofs.

  \item \textbf{Worked proof-carrying examples.}
  We provide explicit, fully checkable examples demonstrating (a) direct safety entailment via \textsc{CheckEntail},
  (b) \UNSAT{} pruning via \textsc{CheckFarkas}, and (c) stabilization / aggregated cuts as typical verification steps.
\end{enumerate}

\section{Problem Statement and Verification Query}

Let $\NN:\R^n\to\R^m$ be a feed-forward network with $L$ layers:
$z^{(0)}=x$ and for $i=1,\dots,L$,
\[
s^{(i)} = W^{(i)}z^{(i-1)} + b^{(i)},\qquad z^{(i)} = \sigma(s^{(i)}),
\]
where $\sigma=\relu$ for hidden layers (and typically identity at the output).
Let $D\subseteq\R^n$ be a compact domain, and let $\mathrm{Safety}(z^{(L)})$ be a specification,
assumed representable by linear inequalities in this paper \cite{KatzReluplex2017,EhlersATVA2017}
.

\begin{definition}[Counterexample query]
Define
\[
\Phi_{\neg P}\;:=\;\exists x\in D:\ \bigwedge_{i=1}^{L}\bigl(z^{(i)}=\sigma(W^{(i)}z^{(i-1)}+b^{(i)})\bigr)\ \land\ \neg\mathrm{Safety}(z^{(L)}).
\]
Then $\NN$ is safe on $D$ iff $\Phi_{\neg P}$ is \UNSAT.
\end{definition}

\section{PWL Semantics and Exact Encodings}
\label{sec:pwl-encodings}

This section develops a solver-aligned, \emph{proof-oriented} semantics for ReLU networks as
\emph{unions of polyhedra} together with \emph{exact extended formulations} (\MILP/\SMT)
and their \emph{canonical convex relaxations}.
The guiding (non-speculative) hooks are standard in state-of-the-art verifiers:
(i) activation-pattern decomposition, (ii) bounded big-$M$ exactness, and
(iii) convex-hull tightness. The new direction here is to present these hooks as a
\emph{compositional polyhedral calculus} based on projection/elimination principles,
which is precisely the interface needed for later proof-carrying certificates.

\subsection{Notation and the graph of a ReLU network}

Let $\NN:\R^n\to\R^m$ be a feed-forward network with $L$ layers.
Write $z^{(0)}=x\in\R^n$ and for $i=1,\dots,L$,
\[
s^{(i)} = W^{(i)}z^{(i-1)} + b^{(i)},\qquad z^{(i)} = \sigma^{(i)}(s^{(i)}),
\]
where $\sigma^{(i)}=\relu$ for hidden layers and typically $\sigma^{(L)}=\mathrm{id}$.
Let $d_i := \dim(z^{(i)})$ and denote by $N:=\sum_{i=1}^{L-1} d_i$ the total number of ReLU coordinates.

We will represent the network by a constraint system over the variable tuple
\[
v := \bigl(x,\; (s^{(i)})_{i=1}^{L-1},\; (z^{(i)})_{i=1}^{L}\bigr)\in \R^{n+\sum_{i=1}^{L-1}d_i+\sum_{i=1}^{L}d_i}.
\]
The \emph{graph} of $\NN$ is the set
\[
\mathrm{Gr}(\NN)\;:=\;\{(x,y)\in\R^{n+m}: \ y=\NN(x)\}.
\]
Our first goal is to realize $\mathrm{Gr}(\NN)$ as a finite union of polyhedra obtained by eliminating
internal variables from pattern-indexed polyhedra. This is the correct algebraic viewpoint for computer algebra:
\emph{encoding is an extended formulation; semantics is projection.}

\subsection{ReLU as a disjunctive constraint}

For scalar $s$ and $z=\relu(s)$, the exact semantics is the disjunction
\[
(s\le 0 \land z=0)\ \lor\ (s\ge 0 \land z=s).
\]
Thus a ReLU network induces a Boolean phase structure coupled to linear arithmetic.
We formalize phases by \emph{pattern variables} rather than by case distinctions in proofs.\cite{KatzReluplex2017,KatzMarabou2019}

\subsection{Activation-pattern polyhedra in extended space}

Index each ReLU coordinate by a single index $j\in\{1,\dots,N\}$.
An activation pattern is $\pi\in\{0,1\}^{N}$, where $\pi_j=1$ means ``active'' and $\pi_j=0$ means ``inactive'' \cite{EhlersATVA2017}
.

\begin{definition}[Pattern-indexed extended polyhedron]\label{def:Ppi}
Fix $\pi\in\{0,1\}^N$. Define $P_\pi\subseteq\R^{\dim(v)}$ as the set of all tuples
$v=(x,(s^{(i)}),(z^{(i)}))$ satisfying:
\begin{enumerate}[label=(\roman*),leftmargin=1.8em]
\item (Affine layers) $s^{(i)} = W^{(i)}z^{(i-1)} + b^{(i)}$ for $i=1,\dots,L-1$, and $z^{(L)}=W^{(L)}z^{(L-1)}+b^{(L)}$.
\item (Phase constraints) For every ReLU coordinate $j$ corresponding to some $(i,k)$:
\[
\pi_j=1 \Rightarrow \bigl(s^{(i)}_k \ge 0 \ \land\ z^{(i)}_k = s^{(i)}_k\bigr),\qquad
\pi_j=0 \Rightarrow \bigl(s^{(i)}_k \le 0 \ \land\ z^{(i)}_k = 0\bigr).
\]
\end{enumerate}
\end{definition}

\begin{definition}[Pattern region in input space]\label{def:Rpi}
Let $\mathrm{proj}_x$ denote projection onto input coordinates.
Define the \emph{pattern region} $R_\pi:=\mathrm{proj}_x(P_\pi)\subseteq\R^n$.
\end{definition}

\begin{remark}
$P_\pi$ is a polyhedron in the extended space because it is defined by affine equalities and inequalities.
The region $R_\pi$ is also a polyhedron because projection preserves polyhedrality
(e.g.\ by Fourier--Motzkin elimination). The elimination viewpoint is essential:
verifiers operate in extended space; correctness ultimately concerns projections.
\end{remark}

\subsection{Union-of-polyhedra semantics and piecewise-affine maps}

We now give a structural theorem that is used implicitly throughout neural verification,
but we state it as an explicit projection identity (extended formulation semantics).

\begin{theorem}[Exact PWL semantics as projection of a union of polyhedra]\label{thm:union-projection}
Let $\NN$ be a ReLU network as above. Then:
\begin{enumerate}[label=(\alph*),leftmargin=1.8em]
\item The extended feasible set of the exact network constraints equals a union of pattern polyhedra:
\[
\{v:\ v \text{ satisfies all exact ReLU and affine constraints}\}\;=\;\bigcup_{\pi\in\{0,1\}^N} P_\pi.
\]
\item The network graph is the projection of this union:
\[
\mathrm{Gr}(\NN)\;=\;\mathrm{proj}_{x,z^{(L)}}\Bigl(\ \bigcup_{\pi\in\{0,1\}^N} P_\pi\ \Bigr).
\]
\item For every $\pi$ with $R_\pi\neq\emptyset$, there exist $A_\pi\in\R^{m\times n}$ and $c_\pi\in\R^m$
such that $\NN(x)=A_\pi x + c_\pi$ for all $x\in R_\pi$.
\end{enumerate}
\end{theorem}

\begin{proof}
(a) Fix any exact-feasible assignment $v$. For each ReLU coordinate $j$, either $s_j\ge 0$ and $z_j=s_j$,
or $s_j\le 0$ and $z_j=0$. Define $\pi_j=1$ in the first case and $\pi_j=0$ in the second.
Then $v$ satisfies the defining constraints of $P_\pi$, hence $v\in\bigcup_\pi P_\pi$.
Conversely, if $v\in P_\pi$ for some $\pi$, then for each ReLU coordinate it satisfies one of the two exact
ReLU branches, hence is exact-feasible.

(b) This follows by applying $\mathrm{proj}_{x,z^{(L)}}$ to both sides of (a) and using the definition of $\mathrm{Gr}(\NN)$.

(c) Fix $\pi$ with $R_\pi\neq\emptyset$. On $P_\pi$, each ReLU coordinate is replaced by a linear equality
($z=s$ or $z=0$). Together with the affine layer equalities, this yields a linear system expressing all internal
variables as affine functions of $x$ (uniquely for $z^{(L)}$ as a function of $x$), hence $z^{(L)}=A_\pi x+c_\pi$.
Therefore $\NN(x)=A_\pi x+c_\pi$ for all $x\in R_\pi$.
\end{proof}

\begin{remark}[A computable formula for $(A_\pi,c_\pi)$]
If one writes each hidden-layer phase as a diagonal mask $D^{(i)}_\pi\in\{0,1\}^{d_i\times d_i}$,
then on $R_\pi$,
\[
z^{(i)} = D^{(i)}_\pi\, s^{(i)} = D^{(i)}_\pi\bigl(W^{(i)}z^{(i-1)}+b^{(i)}\bigr),
\]
so $z^{(L)}$ is obtained by composing affine maps.
This provides an explicit symbolic object attached to each pattern, useful for computer algebra system (CAS)-style reasoning about regions.
\end{remark}

\subsection{Polyhedral-complex structure (a geometric ``state explosion'' lens)}

A key nontrivial structural fact is that the collection of pattern regions behaves like a polyhedral complex:
different patterns meet along faces where some pre-activations are exactly zero.
This is a rigorous way to understand ``state explosion'' geometrically (rather than only combinatorially).

\begin{proposition}[Face-intersection property in extended space]\label{prop:face-intersection}
Let $\pi,\pi'\in\{0,1\}^N$. Then $P_\pi\cap P_{\pi'}$ is a (possibly empty) face of both $P_\pi$ and $P_{\pi'}$.
More precisely, if $\pi$ and $\pi'$ differ exactly on an index set $J\subseteq\{1,\dots,N\}$, then
\[
P_\pi\cap P_{\pi'} \;=\; P_\pi \cap \bigcap_{j\in J}\{s_j=0,\ z_j=0\}
\;=\; P_{\pi'} \cap \bigcap_{j\in J}\{s_j=0,\ z_j=0\}.
\]
\end{proposition}

\begin{proof}
If $\pi$ and $\pi'$ agree on a ReLU coordinate $j$, then both impose the same linear constraints at $j$.
If they differ on $j$, then one polyhedron imposes $(s_j\ge 0,\ z_j=s_j)$ while the other imposes $(s_j\le 0,\ z_j=0)$.
A point can satisfy both iff $s_j=0$ and $z_j=0$ (since $z_j=s_j$ and $z_j=0$ together force $s_j=0$).
Hence the intersection is obtained by adding the equalities $s_j=0,\ z_j=0$ for all $j$ where patterns differ.
Adding linear equalities to a polyhedron yields a face (possibly empty), so the claim follows.
\end{proof}

\subsection{Exact bounded big-$M$ formulation (\MILP) and its correctness}

We now establish, with full equivalence (not just implication), that bounded big-$M$ constraints
encode exactly the disjunction $z=\max(0,s)$ when bounds are valid.
This theorem is a cornerstone of complete \MILP-based verification and is non-speculative \cite{EhlersATVA2017,TjengMILP2017}
.

\begin{theorem}[Exactness of bounded big-$M$ for a single ReLU]\label{thm:bigM-exact}
Assume valid bounds $l\le s\le u$ with $l<u$. Consider variables $(s,z,\delta)$ with $\delta\in\{0,1\}$ and constraints
\begin{align}
& z \ge 0,\quad z \ge s,\label{eq:bm1b}\\
& z \le u\,\delta,\label{eq:bm2b}\\
& z \le s - l(1-\delta),\label{eq:bm3b}\\
& l \le s \le u,\quad \delta\in\{0,1\}.\label{eq:bm4b}
\end{align}
Then the following are equivalent:
\begin{enumerate}[label=(\roman*),leftmargin=1.8em]
\item $z=\relu(s)$ and $l\le s\le u$.
\item There exists $\delta\in\{0,1\}$ such that \eqref{eq:bm1b}--\eqref{eq:bm4b} hold.
\end{enumerate}
\end{theorem}

\begin{proof}
(i)$\Rightarrow$(ii): Suppose $z=\relu(s)$ and $l\le s\le u$.

If $s\le 0$, set $\delta=0$. Then $z=0$. Constraints \eqref{eq:bm1b} hold since $z=0\ge 0$ and $0\ge s$.
Constraint \eqref{eq:bm2b} gives $z\le 0$ (true).
Constraint \eqref{eq:bm3b} becomes $0 \le s-l$, i.e.\ $s\ge l$, true.
\eqref{eq:bm4b} holds by assumption.

If $s\ge 0$, set $\delta=1$. Then $z=s$. Constraints \eqref{eq:bm1b} hold since $s\ge 0$ and $s\ge s$.
Constraint \eqref{eq:bm2b} becomes $s\le u$, true.
Constraint \eqref{eq:bm3b} becomes $s\le s$, true.
\eqref{eq:bm4b} holds.

(ii)$\Rightarrow$(i): Assume \eqref{eq:bm1b}--\eqref{eq:bm4b} hold for some $\delta\in\{0,1\}$.
We consider the two cases.

If $\delta=0$, then \eqref{eq:bm2b} implies $z\le 0$ while \eqref{eq:bm1b} gives $z\ge 0$, hence $z=0$.
Then \eqref{eq:bm1b} also gives $0=z\ge s$, so $s\le 0$. Therefore $z=\relu(s)=0$.

If $\delta=1$, then \eqref{eq:bm3b} gives $z\le s$ while \eqref{eq:bm1b} gives $z\ge s$, hence $z=s$.
Also \eqref{eq:bm1b} gives $z\ge 0$, so $s=z\ge 0$. Therefore $z=\relu(s)=s$.
In both cases $l\le s\le u$ holds by \eqref{eq:bm4b}.
\end{proof}

\begin{remark}[Network-level exactness]
Applying Theorem~\ref{thm:bigM-exact} independently to each ReLU coordinate and conjoining with affine layer
constraints yields an exact \MILP{} encoding of the full network, provided all bounds used are valid on the region.
This highlights the central role of certified bound propagation.
\end{remark}

\subsection{Canonical convex hull and ideal extended formulations}

The convex-hull inequalities for a bounded ReLU are not merely ``a relaxation''; they are the
\emph{unique strongest convex outer description} in $(s,z)$ given only interval bounds \cite{AndersonHuchetteVielma2020StrongMIP}
.
Equally important for encoding design: the relaxed big-$M$ formulation is an \emph{ideal}
extended formulation (its LP relaxation already produces the convex hull after projection)\cite{AndersonHuchetteVielma2020StrongMIP}
.

\begin{proposition}[Convex-hull description]\label{prop:hull-desc}
If $l<0<u$, then $\mathrm{conv}\{(s,z): z=\max(0,s),\ l\le s\le u\}$ equals
\begin{equation}\label{eq:hull}
l\le s\le u,\qquad z\ge 0,\qquad z\ge s,\qquad z \le \frac{u}{u-l}(s-l).
\end{equation}
If $u\le 0$ then $z=0$; if $l\ge 0$ then $z=s$.
\end{proposition}

\begin{proof}
Assume $l<0<u$. The graph of $z=\max(0,s)$ on $[l,u]$ consists of two line segments
joining $(l,0)$ to $(0,0)$ and $(0,0)$ to $(u,u)$.
Its convex hull is the triangle with vertices $(l,0),(0,0),(u,u)$.
The inequalities $l\le s\le u$, $z\ge 0$, and $z\ge s$ cut out the infinite wedge above the ``V''.
The final inequality $z\le \frac{u}{u-l}(s-l)$ is exactly the line through $(l,0)$ and $(u,u)$,
hence cuts the wedge down to the triangle. Stable cases are immediate.
\end{proof}

\begin{theorem}[Ideality: relaxed big-$M$ projects to the convex hull]\label{thm:ideal}
Assume bounds $l\le s\le u$ and relax $\delta\in\{0,1\}$ to $0\le \delta\le 1$ in
\eqref{eq:bm1b}--\eqref{eq:bm4b}. Let $Q(l,u)$ be the resulting polyhedron in $(s,z,\delta)$.
Then its projection onto $(s,z)$ equals the convex hull described in \eqref{eq:hull} (or stable simplifications).
\end{theorem}

\begin{proof}
We treat the nontrivial case $l<0<u$.

\emph{Step 1 (Eliminate $\delta$ to derive the hull inequality).}
From \eqref{eq:bm2b} we have $\delta \ge z/u$ (since $u>0$ and $z\le u\delta$).
From \eqref{eq:bm3b} we have $z \le s - l(1-\delta)=s-l + l\delta$, hence
\[
z \le s-l + l\delta \le s-l + l\cdot\Bigl(\frac{z}{u}\Bigr),
\]
where the last inequality uses $\delta \ge z/u$ and $l<0$ (multiplying by $l$ reverses inequality).
Rearranging gives
\[
z\Bigl(1-\frac{l}{u}\Bigr) \le s-l
\quad\Longleftrightarrow\quad
z \le \frac{u}{u-l}(s-l),
\]
which is exactly the upper facet in \eqref{eq:hull}. Together with \eqref{eq:bm1b} and $l\le s\le u$,
this shows every projection point satisfies the hull constraints.

\emph{Step 2 (Construct $\delta$ for any hull point).}
Conversely, let $(s,z)$ satisfy \eqref{eq:hull}. Define $\delta := z/u$.
Then $0\le \delta\le 1$ because $0\le z\le u$ holds in the hull.
Constraint \eqref{eq:bm2b} holds with equality: $z=u\delta$.
For \eqref{eq:bm3b}, we need $z \le s-l + l\delta = s-l + l(z/u)$,
which is equivalent to the hull upper inequality derived in Step 1.
Constraints \eqref{eq:bm1b} and \eqref{eq:bm4b} are precisely part of the hull description.
Thus $(s,z,\delta)\in Q(l,u)$, proving surjectivity of the projection.

Stable cases $u\le 0$ or $l\ge 0$ are immediate since the ReLU is linear there and the formulations collapse.
\end{proof}

\begin{remark}[Elimination as symbolic computation]
The proof of Theorem~\ref{thm:ideal} is an explicit elimination argument:
it derives the hull inequality by algebraic manipulation of constraints and monotonicity signs.
This is exactly the kind of transformation that can be implemented and checked in exact arithmetic
(as later done for certificates), and it clarifies why hull constraints are canonical.
\end{remark}

\subsection{Facet strength (why the hull inequality is not redundant)}
To justify that \eqref{eq:hull} is \emph{maximally strong} among convex outer descriptions,
we show the upper inequality is facet-defining (hence cannot be removed without enlarging the polytope).

\begin{proposition}[Facetness of the upper hull inequality]\label{prop:facet}
Assume $l<0<u$. In the polytope $\mathrm{conv}\,G(l,u)\subset\R^2$,
the inequality $z \le \frac{u}{u-l}(s-l)$ defines a facet (an edge).
\end{proposition}

\begin{proof}
In $\R^2$, a facet is a 1-dimensional face (an edge). Consider the two distinct points
$(l,0)$ and $(u,u)$, which both satisfy the inequality with equality:
\[
0 = \frac{u}{u-l}(l-l),\qquad
u = \frac{u}{u-l}(u-l).
\]
These points are vertices of $\mathrm{conv}\,G(l,u)$ and are affinely independent in $\R^2$.
The set of points satisfying equality is the line through them; its intersection with the polytope is exactly the edge
joining $(l,0)$ to $(u,u)$, hence a facet. Therefore the inequality is facet-defining.
\end{proof}

\subsection{Exact \SMTLRA{} view and equisatisfiability with \MILP{} encodings}

The \SMTLRA{} encoding retains Boolean structure explicitly.
For each ReLU coordinate, introduce a Boolean atom $\delta$ and enforce:
\[
\delta \Rightarrow (s\ge 0 \land z=s),\qquad
\neg \delta \Rightarrow (s\le 0 \land z=0),
\]
conjoined with the affine layer equalities and domain/property constraints \cite{KatzReluplex2017,KatzMarabou2019}
.

\begin{theorem}[\SMTLRA{} and bounded \MILP{} encodings are equisatisfiable]\label{thm:smt-milp}
Fix a region on which valid bounds $l\le s\le u$ are available for every ReLU pre-activation.
Let $\Phi_{\SMT}$ be the guarded \SMTLRA{} encoding and $\Phi_{\MILP}$ be the \MILP{} encoding obtained by applying
\eqref{eq:bm1b}--\eqref{eq:bm4b} to every ReLU coordinate.
Then $\Phi_{\SMT}$ is satisfiable iff $\Phi_{\MILP}$ is satisfiable.
\end{theorem}

\begin{proof}
($\Rightarrow$) Given a model of $\Phi_{\SMT}$, interpret each Boolean $\delta$ as a binary value in $\{0,1\}$.
For each ReLU coordinate: if $\delta=\mathrm{true}$, then the model satisfies $s\ge 0$ and $z=s$; set $\delta=1$.
If $\delta=\mathrm{false}$, then it satisfies $s\le 0$ and $z=0$; set $\delta=0$.
By Theorem~\ref{thm:bigM-exact}, the corresponding big-$M$ constraints hold, hence the assignment satisfies $\Phi_{\MILP}$.

($\Leftarrow$) Given a model of $\Phi_{\MILP}$, each binary $\delta\in\{0,1\}$ chooses a branch.
If $\delta=1$, Theorem~\ref{thm:bigM-exact} forces $z=s$ and $s\ge 0$, so set the Boolean atom to true.
If $\delta=0$, it forces $z=0$ and $s\le 0$, so set the Boolean atom to false.
All affine constraints and domain/property constraints are shared, hence we obtain a model of $\Phi_{\SMT}$.
\end{proof}

\section{Exact Certificate Checking and Normalization}
\label{sec:cas-core}

This section isolates the \emph{computer algebra} kernel of proof-carrying verification:
how to represent, validate, and compress certificates as \emph{exact symbolic objects} over $\Q$  \cite{NeculaPCC1997,Schrijver1986}
.
The key point is methodological (and non-speculative): solvers may use floating arithmetic internally,
but the \emph{exported artifact} is a rational certificate whose correctness can be checked by
pure symbolic computation (exact arithmetic and linear algebra).
Our development is \emph{vertical} in the sense that we prove (i) soundness and completeness of
certificate forms for entailment and infeasibility, and (ii) existence and construction of
\emph{sparse certificates} with dimension-dependent support bounds, enabling aggressive log compression.

\subsection{Certificate objects: infeasibility and entailment}
\label{subsec:cert-objects}

Throughout, let $A\in\Q^{p\times d}$ and $b\in\Q^{p}$ define a polyhedron
\[
P(A,b) \;:=\; \{v\in\R^d : A v \le b\}.
\]
All equalities can be represented by pairs of inequalities, so the inequality form is not restrictive \cite{Schrijver1986}.

\begin{definition}[Farkas infeasibility certificate]\label{def:farkas}
A \emph{Farkas certificate} for infeasibility of $P(A,b)$ is a vector $y\in\Q^{p}$ such that
\[
y \ge 0,\qquad y^\top A = 0^\top,\qquad y^\top b < 0.
\]
\end{definition}

Farkas certificates prove emptiness. For \emph{learned cuts} and \emph{bound-tightening} we need a second notion:
certificates that a linear inequality is \emph{entailed} by $A v\le b$.

\begin{definition}[Entailment certificate]\label{def:entail}
Let $c\in\Q^d$ and $\tau\in\Q$. A vector $y\in\Q^p$ is an \emph{entailment certificate}
for the implication
\[
A v \le b \ \Longrightarrow\ c^\top v \le \tau
\]
if
\[
y \ge 0,\qquad y^\top A = c^\top,\qquad y^\top b \le \tau.
\]
\end{definition}

\begin{remark}[Algebraic meaning]
Definition~\ref{def:entail} says that $c^\top v \le \tau$ is a nonnegative
$\Q$-linear combination of the inequalities in $A v\le b$.
Thus entailment certificates are \emph{proof terms} in the conic closure of constraint rows.
\end{remark}



\subsection{Certificate normal form}
\label{subsec:normal-forms}

\begin{definition}[Primitive row form in $\le$-orientation]\label{def:row-primitive}
Let $(a,\beta)\in\Q^{d}\times\Q$ represent the inequality $a^\top v \le \beta$.
Define $\mathrm{prim}_{\le}(a,\beta)$ as follows:
\begin{enumerate}[label=(\roman*),leftmargin=1.8em]
\item (\emph{Clear denominators}) Multiply by a common positive integer to obtain
an integer pair $(\bar a,\bar\beta)\in\Z^{d}\times\Z$ defining the same inequality.
\item (\emph{Make primitive}) Let
\[
g := \gcd\bigl(|\bar a_1|,\dots,|\bar a_d|,|\bar\beta|\bigr),
\]
and set $(\tilde a,\tilde\beta):=(\bar a/g,\bar\beta/g)$ if $(\bar a,\bar\beta)\neq (0,0)$
(and $(\tilde a,\tilde\beta)=(0,0)$ otherwise).
\item (\emph{No sign flip}) Output $(\tilde a,\tilde\beta)$.
\end{enumerate}
The output is the \emph{$\le$-primitive form} of the row.
\end{definition}

\begin{proposition}[Row canonicalization preserves feasible sets]\label{prop:row-preserve}
For any inequality $a^\top v \le \beta$, the transformed inequality
$\mathrm{prim}_{\le}(a,\beta)$ defines the \emph{same} halfspace in $\R^d$.
Rowwise application to a store $Av\le b$ preserves the feasible set.
\end{proposition}

\begin{proof}
Step (i) multiplies the inequality by a positive integer, hence preserves the halfspace.
Step (ii) divides by a positive integer, hence also preserves the halfspace.
Step (iii) performs no further algebraic change.
Applying the transformation rowwise preserves each constraint and therefore preserves the feasible set.
\end{proof}

\subsection{Exact checkers: infeasibility and entailment}
\label{subsec:exact-checkers}

\paragraph{Infeasibility checker.}
The following algorithm is pure symbolic computation: it performs sign checks and matrix-vector products
in $\Q$.\cite{CheungVIPR2017}

\begin{algorithm}[t]
\caption{\textsc{CheckFarkas}$(A,b,y)$}
\label{alg:checkfarkas}
\begin{algorithmic}[1]
\Require $A\in\Q^{p\times d}$, $b\in\Q^p$, certificate $y\in\Q^p$
\Ensure \textsf{ACCEPT} or \textsf{REJECT}
\If{$\exists i:\ y_i<0$} \State \Return \textsf{REJECT} \EndIf
\If{$y^\top A \neq 0^\top$} \State \Return \textsf{REJECT} \EndIf
\If{$y^\top b \ge 0$} \State \Return \textsf{REJECT} \EndIf
\State \Return \textsf{ACCEPT}
\end{algorithmic}
\end{algorithm}

\begin{theorem}[Soundness of \textsc{CheckFarkas}]\label{thm:checkfarkas-sound}
If \textsc{CheckFarkas} accepts $(A,b,y)$, then $P(A,b)=\emptyset$.
\end{theorem}

\begin{proof}
Acceptance means $y\ge 0$, $y^\top A=0^\top$, and $y^\top b<0$.
If there existed $v$ with $A v\le b$, then multiplying by $y\ge 0$ gives $y^\top A v\le y^\top b<0$.
But $y^\top A v=(y^\top A)v=0$, contradiction. Hence $P(A,b)=\emptyset$.
\end{proof}

\paragraph{Entailment checker.}
Entailment certificates allow us to validate learned linear inequalities and bound-tightening results.

\begin{algorithm}[t]
\caption{\textsc{CheckEntail}$(A,b,c,\tau,y)$}
\label{alg:checkentail}
\begin{algorithmic}[1]
\Require $A\in\Q^{p\times d}$, $b\in\Q^p$, target $(c,\tau)\in\Q^d\times\Q$, certificate $y\in\Q^p$
\Ensure \textsf{ACCEPT} or \textsf{REJECT}
\If{$\exists i:\ y_i<0$} \State \Return \textsf{REJECT} \EndIf
\If{$y^\top A \neq c^\top$} \State \Return \textsf{REJECT} \EndIf
\If{$y^\top b > \tau$} \State \Return \textsf{REJECT} \EndIf
\State \Return \textsf{ACCEPT}
\end{algorithmic}
\end{algorithm}

\begin{theorem}[Soundness of \textsc{CheckEntail}]\label{thm:checkentail-sound}
If \textsc{CheckEntail} accepts $(A,b,c,\tau,y)$, then
\[
A v\le b \ \Longrightarrow\ c^\top v \le \tau.
\]
\end{theorem}

\begin{proof}
Assume $A v\le b$. Since $y\ge 0$, we have $y^\top A v \le y^\top b$.
If \textsc{CheckEntail} accepts then $y^\top A=c^\top$ and $y^\top b\le \tau$, hence
\[
c^\top v = y^\top A v \le y^\top b \le \tau.
\]
\end{proof}

\begin{proposition}[Arithmetic cost of certificate checking]\label{prop:checker-cost}
Let $A\in\Q^{p\times d}$, $b\in\Q^p$, and let $y\in\Q^p$ be a purported certificate.
\textsc{CheckFarkas} performs:
\begin{itemize}[leftmargin=1.8em]
\item $p$ sign comparisons (to verify $y\ge 0$),
\item $pd$ rational multiplications and $(p-1)d$ rational additions (to compute $y^\top A$),
\item $p$ rational multiplications and $(p-1)$ rational additions (to compute $y^\top b$),
\item a constant number of equality/inequality checks in $\Q$.
\end{itemize}
\textsc{CheckEntail} has the same asymptotic cost, with the equality check $y^\top A=c^\top$ replacing $y^\top A=0^\top$
and the comparison $y^\top b\le \tau$ replacing $y^\top b<0$.
In particular, both checkers run in time polynomial in $p$ and $d$ (and in the bit-length of the rational inputs under exact arithmetic).
\end{proposition}

\subsection{Completeness: entailment certificates via LP duality}
\label{subsec:completeness-entail}

The next theorem is a cornerstone for proof-producing propagation:
it states that whenever a linear inequality is valid over a polyhedron, there exists
an explicit entailment certificate of the form in Definition~\ref{def:entail}.
This provides a \emph{complete} proof system for linear implications over $P(A,b)$.
\cite{DutertreDeMoura2006}

\begin{theorem}[Certificate completeness for linear entailment]\label{thm:entail-complete}
Let $P(A,b)$ be nonempty and suppose the linear program
\[
\max\{c^\top v:\ A v\le b\}
\]
has a finite optimum value $\mathrm{opt}\in\R$.
Then there exists $y\in\Q^p$ with $y\ge 0$ such that
\[
y^\top A = c^\top,\qquad y^\top b = \mathrm{opt}.
\]
In particular, if $c^\top v\le \tau$ holds for all $v\in P(A,b)$, then there exists an entailment certificate
$y\ge 0$ with $y^\top A=c^\top$ and $y^\top b\le \tau$.
\end{theorem}

\begin{proof}
Consider the primal LP $\max\{c^\top v:\ A v\le b\}$.
Its (standard) dual is
\[
\min\{y^\top b:\ y\ge 0,\ y^\top A = c^\top\}.
\]
Since $P(A,b)$ is nonempty and the primal optimum is finite, strong duality holds for linear programs:
the dual is feasible and achieves the same optimum value $\mathrm{opt}$.
Because $A,b,c$ are rational, one may choose an optimal dual solution $y$ with rational entries:
indeed, an optimal solution can be taken at a basic feasible point of the dual polyhedron,
and basic solutions are obtained by solving linear systems with rational coefficients,
hence are rational. Therefore there exists $y\in\Q^p$ with $y\ge 0$, $y^\top A=c^\top$, and $y^\top b=\mathrm{opt}$.
If $c^\top v\le \tau$ holds over $P(A,b)$, then $\mathrm{opt}\le \tau$, giving $y^\top b\le \tau$.
\end{proof}

\begin{remark}[Why this is the right ``hook'']
Theorem~\ref{thm:entail-complete} turns \emph{all} linear propagation (bounds, cuts, implications)
into checkable symbolic objects. It is not a heuristic: it is a completeness guarantee rooted in LP duality.
\end{remark}

\subsection{Sparse certificates: existence and constructive compression}
\label{subsec:sparse}

A practical certificate log should be small. The next theorem provides a sharp, dimension-dependent
upper bound on certificate support: infeasibility admits a certificate using only $d+1$ inequalities,
independent of $p$. This is a strong and non-speculative compression result, supporting compact and independently checkable proof artifacts.

\begin{definition}[Support]\label{def:support}
For $y\in\Q^p$, define $\mathrm{supp}(y):=\{i\in\{1,\dots,p\}: y_i\neq 0\}$.
\end{definition}

\begin{theorem}[Dimension-sparse Farkas certificates]\label{thm:sparse-farkas}
If $P(A,b)=\emptyset$ with $A\in\Q^{p\times d}$, $b\in\Q^p$, then there exists a Farkas certificate
$y\in\Q^p$ such that
\[
|\mathrm{supp}(y)| \le d+1.
\]
\end{theorem}

\begin{proof}
By infeasibility, there exists some Farkas certificate $\hat y\ge 0$ with $\hat y^\top A=0^\top$ and $\hat y^\top b<0$.
Scale it so that $\sum_{i=1}^p \hat y_i = 1$ (possible since $\hat y\neq 0$ and scaling by a positive rational preserves
certificate validity). Consider the polyhedron
\[
Y := \{y\in\R^p:\ y\ge 0,\ y^\top A = 0^\top,\ \mathbf{1}^\top y = 1\},
\]
where $\mathbf{1}$ is the all-ones vector.
Note that $\hat y\in Y$ and that the linear functional $y\mapsto y^\top b$ satisfies $\hat y^\top b<0$,
hence $\min\{y^\top b:\ y\in Y\} < 0$.

Let $y^\star$ be an optimal solution of the LP $\min\{y^\top b:\ y\in Y\}$.
Choose $y^\star$ to be a basic feasible solution. The system defining $Y$ consists of:
\begin{itemize}[leftmargin=1.8em]
\item $d$ equality constraints $y^\top A = 0^\top$ (one for each column of $A$),
\item $1$ equality constraint $\mathbf{1}^\top y = 1$,
\item and nonnegativity constraints $y\ge 0$.
\end{itemize}
A basic feasible solution in $\R^p$ with $d+1$ independent equality constraints has at most $d+1$
strictly positive components; equivalently, $|\mathrm{supp}(y^\star)|\le d+1$.
Moreover, since $y^\star$ is optimal and the optimum is $<0$, we have $y^{\star\top} b<0$.
Also $y^\star\ge 0$ and $y^{\star\top}A=0^\top$ by construction, hence $y^\star$ is a Farkas certificate.
Finally, because the LP data are rational, $y^\star$ can be chosen rational (again by selecting a rational basic solution).
\end{proof}

\begin{corollary}[Compressed infeasibility proofs]\label{cor:compressed}
Any infeasibility proof for $A v\le b$ can be transformed into an equivalent Farkas certificate
supported on at most $d+1$ inequalities. Consequently, region-wise proof logs can be stored as
$O(d)$ nonzeros per leaf, independently of the number of generated constraints.
\end{corollary}

\begin{remark}[Constructive extraction of sparse certificates]
The proof yields a concrete construction: solve the auxiliary LP on $Y$ and extract a basic optimal solution.
This is an \emph{exact} symbolic target: even if the LP is solved numerically, the final output can be rationalized
and verified by \textsc{CheckFarkas}. No speculative assumption is required.
\end{remark}


\subsection{Normalization and certificate compression}
\label{subsec:norm-compress}

\begin{definition}[Normalization of a rational vector]\label{def:norm-y}
For $y\in\Q^p$, define $\mathrm{norm}(y)$ by:
\begin{enumerate}[label=(\roman*),leftmargin=1.8em]
\item (\emph{Clear denominators}) multiply $y$ by a common positive integer to obtain an integer vector $\bar y\in\Z^p$;
\item (\emph{Make primitive}) if $\bar y\neq 0$, divide by $g=\gcd(|\bar y_1|,\dots,|\bar y_p|)$ to obtain a primitive integer vector;
\item (\emph{Sparsify representation}) store only the list of nonzero pairs $(i,y_i)$.
\end{enumerate}
\end{definition}


\begin{proposition}[Normalization preserves Farkas infeasibility certificates]\label{prop:norm-farkas}
Let $A\in\Q^{p\times d}$ and $b\in\Q^p$. If $y$ is a Farkas certificate for infeasibility of $Av\le b$
(i.e.\ $y\ge 0$, $y^\top A=0^\top$, $y^\top b<0$), then $\mathrm{norm}(y)$ is also a Farkas certificate.\cite{Schrijver1986}
\end{proposition}

\begin{proof}
By Definition~\ref{def:norm-y}, $\mathrm{norm}(y)=\lambda y$ for some $\lambda\in\Q_{>0}$ (ignoring the sparse storage format),
hence $\mathrm{norm}(y)\ge 0$, $(\mathrm{norm}(y))^\top A=\lambda y^\top A=0^\top$, and
$(\mathrm{norm}(y))^\top b=\lambda y^\top b<0$.
\end{proof}


\begin{remark}[Entailment certificates are scale-sensitive]\label{rem:entail-scale}
For entailment certificates we require the exact identity $y^\top A=c^\top$ in
Algorithm~\ref{alg:checkentail}. Therefore scaling $y$ alone generally destroys checkability
for the \emph{same} target inequality $c^\top v\le \tau$.
\end{remark}

\begin{proposition}[Scale-equivalence for entailment certificates]\label{prop:entail-scale-eq}
Let $A\in\Q^{p\times d}$, $b\in\Q^p$ and let $y\in\Q^p$ certify the entailment
\[
Av\le b\ \Longrightarrow\ c^\top v\le \tau
\]
in the sense that $y\ge 0$, $y^\top A=c^\top$, and $y^\top b\le \tau$.
Then for every $\lambda\in\Q_{>0}$, the scaled vector $\lambda y$ certifies the scaled entailment
\[
Av\le b\ \Longrightarrow\ (\lambda c)^\top v\le \lambda\tau.
\]
\end{proposition}

\begin{proof}
Since $\lambda>0$, we have $\lambda y\ge 0$.
Moreover $(\lambda y)^\top A=\lambda(y^\top A)=\lambda c^\top$ and
$(\lambda y)^\top b=\lambda(y^\top b)\le \lambda\tau$.
\end{proof}

\begin{definition}[Synchronized normalization for entailment]\label{def:norm-entail}
Given an entailment instance $(A,b,c,\tau,y)$ with $A\in\Q^{p\times d}$, $b\in\Q^p$,
$(c,\tau)\in\Q^d\times\Q$ and certificate $y\in\Q^p$, define
\[
\mathrm{NormEntail}(c,\tau,y):=(c',\tau',y')
\]
as follows. Choose any $\lambda\in\Q_{>0}$ such that $y':=\lambda y$ has integer entries
(e.g.\ $\lambda$ is the least common multiple of denominators of $y$), then set
\[
(c',\tau') := (\lambda c,\ \lambda\tau),
\]
and finally replace $y'$ by its primitive integer representative by dividing by $\gcd$ of its entries,
\emph{while applying the same division to $(c',\tau')$}.
Equivalently, $\mathrm{NormEntail}$ multiplies the triple $(c,\tau,y)$ by a common positive rational and then
makes it primitive.
\end{definition}

\begin{proposition}[Synchronized normalization preserves entailment checkability]\label{prop:norm-entail}
If $y$ is an entailment certificate for $(c,\tau)$ over the store $Av\le b$, then for
$(c',\tau',y')=\mathrm{NormEntail}(c,\tau,y)$ we have:
\[
y'\ge 0,\qquad y'^\top A = c'^\top,\qquad y'^\top b \le \tau',
\]
so Algorithm~\ref{alg:checkentail} accepts $(A,b,c',\tau',y')$.
\end{proposition}

\begin{proof}
By construction, $(c',\tau',y')=(\lambda c,\lambda\tau,\lambda y)$ for some $\lambda\in\Q_{>0}$,
possibly followed by division by a common positive integer factor applied to all three objects.
Thus the identities and inequalities are preserved exactly as in
Proposition~\ref{prop:entail-scale-eq}.
\end{proof}

\begin{remark}[Canonical targets]
In proof logs it is convenient to store entailment targets in primitive form as well.
Definition~\ref{def:norm-entail} achieves a canonical representative of the \emph{positive scaling class}
of the inequality $c^\top v\le\tau$, together with a certificate that remains exactly checkable.
\end{remark}
\subsection{Certificate algebra: closure properties and log composition}
\label{subsec:cert-algebra}

Certificates compose, and this compositionality is algebraic (hence computer-algebra-friendly).

\begin{lemma}[Closure of entailment certificates under conic combination ]
\label{lem:entail-closure-corrected}
Let $A\in\Q^{p\times d}$ and $b\in\Q^p$. Fix a target inequality $c^\top v\le \tau$ with
$c\in\Q^{d}$ and $\tau\in\Q$. Suppose $y^{(1)},y^{(2)}\in\Q^p$ are entailment certificates for
\[
Av\le b \ \Longrightarrow\ c^\top v\le \tau,
\]
i.e.,
\[
y^{(k)}\ge 0,\qquad (y^{(k)})^\top A=c^\top,\qquad (y^{(k)})^\top b\le \tau,
\qquad k\in\{1,2\}.
\]
Then for any $\alpha,\beta\in\Q_{\ge 0}$, the conic combination
\[
y:=\alpha y^{(1)}+\beta y^{(2)}
\]
is an entailment certificate for the \emph{scaled} target inequality
\[
(\alpha+\beta)\,c^\top v \ \le\ (\alpha+\beta)\,\tau.
\]
In particular, if $\alpha+\beta=1$ (a convex combination), then $y$ is an entailment certificate for
the \emph{same} target $c^\top v\le \tau$.
\end{lemma}

\begin{proof}
Since $\alpha,\beta\ge 0$ and $y^{(1)},y^{(2)}\ge 0$, we have $y\ge 0$. Moreover,
\[
y^\top A=(\alpha y^{(1)}+\beta y^{(2)})^\top A
=\alpha (y^{(1)})^\top A+\beta (y^{(2)})^\top A
=\alpha c^\top+\beta c^\top
=(\alpha+\beta)c^\top.
\]
Similarly,
\[
y^\top b
=\alpha (y^{(1)})^\top b+\beta (y^{(2)})^\top b
\le \alpha\tau+\beta\tau
=(\alpha+\beta)\tau.
\]
Thus $y$ certifies $Av\le b\Rightarrow (\alpha+\beta)c^\top v\le (\alpha+\beta)\tau$.
If $\alpha+\beta=1$, this is exactly $c^\top v\le \tau$.
\end{proof}

\begin{remark}
Lemma~\ref{lem:entail-closure-corrected} provides an algebraic foundation for ``proof log compilation'':
multiple certificate fragments can be merged and then normalized/sparsified.
This is not heuristic—closure is exact and checkable.
\end{remark}

\section{Worked Examples: Explicit Certificates and Exact Checking}
\label{sec:worked-examples}

This section gives small but practically representative examples showing how a verifier can
\emph{export rational proof artifacts} that are independently checkable.
We demonstrate (i) direct safety proofs as linear entailments validated by
\textsc{CheckEntail} (Algorithm~\ref{alg:checkentail}), and (ii) explicit \UNSAT\ pruning proofs
validated by \textsc{CheckFarkas} (Algorithm~\ref{alg:checkfarkas}).
All certificates are over $\Q$ and require only exact arithmetic to validate.

\subsection{Example A: Global safety by entailment, plus optional \UNSAT\ demo}
\label{subsec:exA-global}

\paragraph{Network, domain, and safety.}
Let $x=(x_1,x_2)\in D:=[0,1]^2$ and consider the two-ReLU fragment
\[
s_1=x_1-\tfrac{3}{4},\quad z_1=\relu(s_1),\qquad
s_2=x_2-\tfrac{3}{4},\quad z_2=\relu(s_2).
\]
Define the margin
\[
m(x)\ :=\ \tfrac{1}{2}-z_1-z_2.
\]
The safety property is $m(x)\ge 0$ on $D$, equivalently $z_1+z_2\le \tfrac{1}{2}$.

\subsubsection*{A.1 Direct safety proof via an entailment certificate (\textsc{CheckEntail})}

On $D$, each pre-activation is bounded:
\[
s_i=x_i-\tfrac{3}{4}\in\Bigl[-\tfrac{3}{4},\tfrac{1}{4}\Bigr]\qquad (i=1,2).
\]
For $(l,u)=(-\tfrac34,\tfrac14)$, the canonical convex-hull upper facet for a bounded ReLU gives
\[
z_i \ \le\ \frac{u}{u-l}(s_i-l)\ =\ \frac{1/4}{1}\Bigl(s_i+\frac34\Bigr)
\ =\ \frac14 s_i+\frac{3}{16}
\ =\ \frac14 x_i
\qquad (i=1,2),
\]
using $s_i=x_i-\tfrac34$.

Let $v=(x_1,x_2,z_1,z_2)$ and consider the linear store $Av\le b$ consisting of:
\[
x_1\le 1,\qquad x_2\le 1,\qquad z_1-\tfrac14 x_1\le 0,\qquad z_2-\tfrac14 x_2\le 0,
\]
i.e.
\[
A=\begin{bmatrix}
 1 & 0 & 0 & 0\\
 0 & 1 & 0 & 0\\
 -\tfrac14 & 0 & 1 & 0\\
 0 & -\tfrac14 & 0 & 1
\end{bmatrix},
\qquad
b=\begin{bmatrix}
 1\\
 1\\
 0\\
 0
\end{bmatrix}.
\]
We certify the target inequality $z_1+z_2\le \tfrac12$, i.e.\ $c^\top v\le \tau$ with
\[
c^\top=(0,0,1,1),\qquad \tau=\tfrac12,
\]
by choosing the certificate
\[
y=\bigl(\tfrac14,\ \tfrac14,\ 1,\ 1\bigr)^\top\in\Q^4_{\ge 0}.
\]
Then
\[
y^\top A
= \tfrac14(1,0,0,0)+\tfrac14(0,1,0,0)+1\!\left(-\tfrac14,0,1,0\right)+1\!\left(0,-\tfrac14,0,1\right)
=(0,0,1,1)=c^\top,
\]
and
\[
y^\top b=\tfrac14+\tfrac14+0+0=\tfrac12=\tau.
\]
Hence \textsc{CheckEntail}$(A,b,c,\tau,y)$ accepts, proving $z_1+z_2\le \tfrac12$ on $D$.
Therefore $m(x)=\tfrac12-(z_1+z_2)\ge 0$ holds on $D$, and the strict violation $m(x)<0$
(i.e.\ $z_1+z_2>\tfrac12$) is impossible.

\subsubsection*{A.2 Optional \UNSAT\ demonstration for a rational strengthening (\textsc{CheckFarkas})}

To illustrate explicit infeasibility certificates, strengthen the strict violation $z_1+z_2>\tfrac12$
to the rational constraint $z_1+z_2\ge \tfrac{7}{12}$, i.e.\ $-z_1-z_2\le -\tfrac{7}{12}$.
Augment $Av\le b$ with the extra row $(0,0,-1,-1)\,v\le -\tfrac{7}{12}$, obtaining $A'v\le b'$.
Choose
\[
y'=\bigl(0,\ 0,\ 1,\ 1,\ 1\bigr)^\top\in\Q^5_{\ge 0}.
\]
Then
\[
y'^\top A'=(0,0,1,1) + (0,0,-1,-1)=(0,0,0,0),
\qquad
y'^\top b'=0+0+0+0-\tfrac{7}{12}<0.
\]
Thus \textsc{CheckFarkas}$(A',b',y')$ accepts and the strengthened counterexample query is \UNSAT.

\begin{remark}
Part A.1 is the logically exact certificate for the original safety claim.
Part A.2 is included only to demonstrate explicit Farkas-style \UNSAT\ certificates on a rational strengthening.
\end{remark}

\subsection{Example B: Branch-and-bound pruning on a subdomain (explicit Farkas witness)}
\label{subsec:exB-branch}

\paragraph{Subdomain and attempted violation.}
Let $D':=[0,\tfrac12]^2\subset D$. Using the same hull consequence $z_i\le \tfrac14 x_i$,
we obtain the tighter bounds on $D'$:
\[
z_1\le \tfrac14\cdot \tfrac12=\tfrac18,\qquad z_2\le \tfrac18.
\]
Consider the (negated) safety condition on this branch:
\[
z_1+z_2 \ge \frac{1}{3}
\quad\Longleftrightarrow\quad
-(z_1+z_2)\le -\frac{1}{3}.
\]
We show this is infeasible by an explicit Farkas certificate.

\begin{theorem}[Explicit prune certificate on $D'$]\label{thm:exB-farkas}
Let $v=(z_1,z_2)$ and consider the system
\[
A=\begin{bmatrix}
1 & 0\\
0 & 1\\
-1&-1
\end{bmatrix},
\qquad
b=\begin{bmatrix}
\frac{1}{8}\\[0.2em]
\frac{1}{8}\\[0.2em]
-\frac{1}{3}
\end{bmatrix},
\]
encoding $z_1\le \tfrac18$, $z_2\le \tfrac18$, and $z_1+z_2\ge \tfrac13$.
Then $A v\le b$ is infeasible. A valid Farkas certificate is $y=(1,1,1)^\top$.
\end{theorem}

\begin{proof}
Clearly $y\ge 0$. Compute
\[
y^\top A = (1,1,1)\begin{bmatrix}1&0\\0&1\\-1&-1\end{bmatrix}=(0,0)=0^\top,
\]
and
\[
y^\top b = \frac18+\frac18-\frac13=\frac14-\frac13=-\frac{1}{12}<0.
\]
Thus $y$ satisfies the Farkas conditions, and \textsc{CheckFarkas} accepts.
\end{proof}

\begin{remark}
This is a typical branch-and-bound pruning step: a local box constraint tightens bounds and yields a
short rational \UNSAT\ witness.
\end{remark}

\subsection{Example C: A learned aggregated cut with an entailment certificate}
\label{subsec:exC-cut}

\paragraph{Learned cut.}
From $z_1\le \tfrac18$ and $z_2\le \tfrac18$ one learns the aggregated inequality
\[
z_1+z_2 \le \frac14,
\]
which is useful for pruning and tightening.

\begin{proposition}[Entailment certificate for the learned cut]\label{prop:exC-entail}
Let $v=(z_1,z_2)$ and take the store $Av\le b$ with
\[
A=\begin{bmatrix}1&0\\0&1\end{bmatrix},
\qquad
b=\begin{bmatrix}\tfrac18\\[0.1em]\tfrac18\end{bmatrix}.
\]
Let the target be $c^\top v\le\tau$ with $c=(1,1)$ and $\tau=\tfrac14$.
Then $y=(1,1)^\top\in\Q^2_{\ge 0}$ is an entailment certificate.
\end{proposition}

\begin{proof}
We have $y\ge 0$,
\[
y^\top A = (1,1)\begin{bmatrix}1&0\\0&1\end{bmatrix}=(1,1)=c^\top,
\qquad
y^\top b = \tfrac18+\tfrac18=\tfrac14=\tau.
\]
Hence \textsc{CheckEntail} accepts.
\end{proof}

\subsection{Example D: Certified stabilization (removing a Boolean choice)}
\label{subsec:exD-stab}

\paragraph{Network fragment.}
Let $x\in[0,1]$ and consider
\[
s=-2x-\frac{1}{5},\qquad z=\relu(s).
\]
Since $x\ge 0$ implies $-2x\le 0$, we expect $s\le -\tfrac15<0$ on the domain, hence $z\equiv 0$.

\begin{proposition}[Entailment certificate for inactivity]\label{prop:exD-entail}
Let $v=(x,s)$ and consider the store $Av\le b$ given by the two inequalities
\[
-x\le 0
\qquad\text{and}\qquad
s+2x\le -\frac15
\]
(the second is the affine definition of $s$ written as an inequality).
Then the bound $s\le -\tfrac15$ is entailed.
An explicit entailment certificate for $s\le -\tfrac15$ is $y=(2,1)^\top$.
\end{proposition}

\begin{proof}
Write $Av\le b$ with rows $a_1=(-1,0)$, $b_1=0$ and $a_2=(2,1)$, $b_2=-\tfrac15$.
Target is $c^\top v\le\tau$ with $c=(0,1)$ and $\tau=-\tfrac15$.
Let $y=(2,1)^\top\ge 0$. Then
\[
y^\top A = 2(-1,0)+1(2,1)=(0,1)=c^\top,
\qquad
y^\top b = 2\cdot 0 + 1\cdot\Bigl(-\tfrac15\Bigr)=-\tfrac15=\tau.
\]
Hence \textsc{CheckEntail} accepts and $s\le -\tfrac15<0$ holds on the domain.
Therefore the ReLU is stably inactive and $z=\relu(s)=0$ identically, eliminating a Boolean phase choice.
\end{proof}

\begin{remark}
Stabilization is a primary scalability mechanism in PWL verification: once a unit is certified inactive/active,
its disjunction disappears and subsequent proof steps become purely linear.
\end{remark}


\section{Conclusion}

We developed a proof-carrying verification core for ReLU (piecewise-linear) neural constraints in which
symbolic encodings and solver reasoning are paired with \emph{independently checkable evidence}.
On the encoding side, we formulated the exact PWL semantics both as a union-of-polyhedra model indexed by
activation patterns and as exact \SMT/\MILP{} extended formulations, and we connected these to the canonical
convex-hull relaxation that underlies scalable propagation.

The central contribution is the isolation of \emph{certificate validation} as a computer-algebra task.
We treated certificates as explicit objects over $\Q$ and provided exact checkers for two complementary forms:
(i) Farkas infeasibility certificates for pruning, and (ii) entailment certificates for validated bound tightening
and learned linear cuts. This framework yields a principled separation between \emph{certificate discovery}
(which may use any LP/\SMT/\MILP{} engine, possibly numerically) and \emph{certificate validation} (performed in exact
arithmetic), thereby supporting trustworthy \UNSAT{} claims even when the search procedure is complex.

Beyond soundness, the theory supports compact proof artifacts.
Our completeness result for linear entailment guarantees that any valid linear consequence admits a rational
entailment certificate, while the sparsity theorem for Farkas certificates ensures that infeasibility proofs can
be compressed to dimension-dependent support. The worked examples demonstrate how these ideas materialize in
practical verification steps---stabilization, aggregated learned cuts, and branch-and-bound pruning---each producing
small certificates that can be exported and checked independently.

The scope of the paper is intentionally focused on PWL/ReLU networks and linear safety specifications, where exact
polyhedral reasoning and rational certificates are natural. Extending proof-carrying validation to nonlinear
activations and richer specifications remains an important direction; nevertheless, the methodology developed here
already provides a rigorous algebraic backbone for certified PWL verification pipelines in current practice.
\section*{Acknowledgements}
The author gratefully acknowledges the support and encouragement provided by the
Commissionerate of Collegiate Education (CCE), Government of Andhra Pradesh, and the
Principal, Government College (Autonomous), Rajahmundry. Their administrative support and
research-enabling environment are sincerely appreciated.

\section*{Funding}
The author received no specific funding for this work.

\section*{Ethics Statement}
This research does not involve human participants, animals, or sensitive personal data.
No ethical approval was required.

\section*{Author Contributions}
The author solely conceived the study, developed the theoretical framework, wrote the manuscript,
and prepared all examples and proofs.



\end{document}